\documentclass[preprint,12pt]{elsarticle}

\usepackage{amssymb,amsmath}
\usepackage{amsmath}
\usepackage{psfrag, graphicx}
\usepackage{picinpar}
\usepackage{hyperref}
\usepackage{natbib}
\usepackage{graphics}
\usepackage{graphicx}
\usepackage{epsfig}
\usepackage{amssymb}

\newcommand{\disp}{\ensuremath{\displaystyle}}

\def\eps{\varepsilon}

\def\R {\mathbb{R}}
\def\N {\mathbb{N}}
\def\H {{\mathcal H}}

\def\B {{\mathcal B}}

\def\D {{\mathcal D}}
\def\E {{\mathcal E}}
\def\A {{\mathcal A}}
\def\L {{\mathcal L}}
\def\K {{\mathcal K}}

\def\F {{\mathcal F}}

\def \l {\langle}
\def \r {\rangle}
\def \pt {\partial_t}
\def \ptt {\partial_{tt}}

\newtheorem{theorem}{Theorem}
\newtheorem{lemma}[theorem]{Lemma}
\newdefinition{remark}{Remark}
\newdefinition{definition}{Definition}
\newdefinition{notation}{Notation}
\newproof{proof}{Proof}
\newproof{pot}{Proof of Theorem \ref{thm2}}

\numberwithin{equation}{section}

\journal{Elsevier}

\begin{document}

\begin{frontmatter}



\title{Buckling and longterm dynamics of a nonlinear model for the extensible beam}


\author[rvt]{I.˜Bochicchio\corref{cor1}}
\ead{ibochicchio@unisa.it}

\author[focal]{E.˜Vuk}
\ead{vuk@ing.unibs.it}

\cortext[cor1]{Corresponding author}

\address[rvt]{Dipartimento di Matematica e Informatica, Universit\`a di
Salerno,\\ Via Ponte don Melillo, 84084 Fisciano (SA), Italy }
\address[focal]{Dipartimento di Matematica, Universit\`a di
Brescia,\\ Via Valotti 9, 25133 Brescia, Italy}

\begin{abstract}

This work is focused on the longtime behavior of a non linear
evolution problem
describing the vibrations of an extensible elastic homogeneous
beam resting on a viscoelastic foundation with stiffness $k>0$ and
positive damping constant. Buckling of solutions occurs as the
axial load exceeds the first critical value, $\beta_c$, which
turns out to increase piecewise-linearly with $k$. Under hinged
boundary conditions and for a general axial load $P$, the
existence of a global attractor, along with its characterization,
is proved  by exploiting a previous result on the extensible
viscoelastic beam. As $P\leq\beta_c$, the stability of the
straight position is shown for all values of $k$. But, unlike the
case with null stiffness, the exponential decay of the related
energy is proved if $P<\bar\beta(k)$, where
$\bar\beta(k)\leq\beta_c(k)$ and the equality holds only for small
values of  $k$.
\end{abstract}


\end{frontmatter}


\section{Introduction}

\subsection{The model equation}
In this paper we investigate the longtime behavior of the
following evolution problem
\begin{equation}
\label{PROB}
\begin{cases}
\displaystyle \ptt u+\partial_{xxxx}u
-\left(\beta+\int_0^1|\partial_\xi u(\xi,\cdot)|^2 d \xi\right)\partial_{xx}u=-ku-\delta\pt u+f,\\
u(0,t)=u(1,t)=\partial_{xx}u(0,t)=\partial_{xx}u(1,t)=0,\\
u(x,0)=u_0(x),\quad \pt u(x,0)=u_1(x),
\end{cases}
\end{equation}
in the unknown variable $u=u(x,t):[0,1]\times\R^+\to\R$,
$\R^+=[0,\infty)$, which represents the vertical deflection of the
beam. For every $x\in[0,1]$, $u_0$, $u_1$ are assigned data. The
real function $f=f(x)$ is the (given) lateral static load
distribution and $-ku-\delta\pt u$ represents the (uniform)
lateral action effected by the medium surrounding the beam.
Finally, the parameter $\beta\in\R$ accounts for the axial force
acting in the reference configuration: $\beta>0$ when the beam is
stretched, $\beta<0$ when the beam is compressed. Usually, the
axial load is referred as $P=-\beta$.

The solutions to problem~(\ref{PROB}) describe the transversal
vibrations (in dimensionless variables) of an extensible elastic
beam, which is assumed to have hinged ends and to rest on a
viscoelastic foundation with stiffness $k>0$ and damping constant
$\delta>0$. The geometric nonlinearity which is involved accounts
for midplane stretching due to the elongation of the bar. A
simplified version of this beam model has been adopted to study
the vibration of railway track structures resting on a
viscoelastic soil (see \cite{ACC}). There, the elastic and damping
properties of the rail bed are accounted for by continuously
distributed or closely spaced spring--damper units.

In recent years, an increasing attention was payed on the analysis
of vibrations and post-buckling dynamics of nonlinear beam models,
especially in connection with industrial applications. For a
detailed overview, we refer the reader to \cite{NP} and references
therein. Nowadays the study of this subject has become of
particular relevance in the analysis of micromachined beams
\cite{FW,YN} and microbridges \cite{LG}.


It is worth noting that the static counterpart of
problem~(\ref{PROB}) reduces to
\begin{eqnarray}
\label{STATIC}
\begin{cases}
 u{''''}
-\left(\beta+\int_0^1|\partial_\xi u(\xi,\cdot)|^2 \,d \xi\right)u{''}+ku=f,\\
u(0)=u(1)=u{''}(0)=u{''}(1)=0.
\end{cases}
\end{eqnarray}
Obviously, these steady-state equations does not change in
connection with dynamical models accounting for any kind of
additional damping, due to structural and/or external mechanical
dissipation. When $k\equiv0$ the investigation of the solutions to
(\ref{STATIC}) and their stability, in dependence on $\beta$,
represents a classical {\it nonlinear buckling problem} in the
structural mechanics literature (see, for instance,
\cite{B1,DI1,RM}) which traces back to the pioneer paper by
Woinowsky-Krieger \cite{W}. Numerical solutions for this problem
are available in the literature (see, for instance, \cite{CC}).
Recently, a careful analysis of the corresponding buckled
stationary states and their stability properties was performed in
\cite{NE} for all values of $\beta$. In \cite{CZGP} this analysis
was improved to include a more general nonlinear term and a source
$f$ with a general shape. As far as we know, no similar analysis
of (\ref{STATIC}) for the case $k>0$ is present in the literature.

Neglecting the stiffness of the surrounding medium, $k\equiv0$,
the global dynamics of solutions to problem~(\ref{PROB}) has been
first tackled by Hale \cite{HAL}, who proved the existence of a
global attractor for a general $\beta$, relying on the existence
of a suitable Lyapunov functional. The corresponding problem for
an extensible {\it viscoelastic} beam has been addressed in
\cite{IB}, when $\delta>0$, and in \cite{GPV}, when
$\delta\equiv0$.  In spite of the difficulty which is represented
by the geometric nonlinearity, in all these papers  the existence
of the global attractor, along with its optimal regularity, is
obtained for a general value of $\beta$ by means of an abstract
operator setting. The analysis of the bending motion for an
extensible {\it thermoelastic beam} is even more tangled since the
dissipation is entirely contributed by the heat equation, where
the Fourier heat conduction law is assumed.~Neverthelesss, the
existence of a regular global attractor can be shown even
including into (\ref{PROB}) the rotatory inertia term \cite{GNPP}.

A common feature of previously quoted results is the following: if
$f=0$ the exponential decay of the energy is provided when
$\beta>-\beta_c$, so that the unique null solution is
exponentially stable if $P<\beta_c$. On the contrary, as the axial
load $P$ exceeds $\beta_c$ the straight position loses stability
and the beam buckles. So, when $\beta<-\beta_c$ a finite number of
buckled solutions occurs and the global (exponential) attractor
coincides with the unstable trajectories connecting them. For a
general time-independent source term $f$, the number of buckled
solutions may be infinite and the attractor coincides with the
unstable set of the stationary points. The positive critical value
$\beta_c$ is named {\it Euler buckling load} and, in the purely
mechanical case, it is equal to the square root of the first
eingenvalue of the $\partial_{xxxx}$ operator (which is referred
as $\lambda_1$ in the sequel). In the thermoelastic case, because
of the thermal expansion, the mean axial temperature of  the beam
also affects the value of $\beta_c$ (see \cite{GNPP}).

\subsection{Outline of the paper}

At a first sight, problem~(\ref{PROB}) with $k>0$ looks like a
slight modification of previously scrutinized models where $k$
vanishes. This is partially true. In particular, we remark that
the restoring elastic force, acting on each point of the beam,
opposes the buckling phenomenon.  So, the Euler buckling limit
$\beta_c$ is no longer equal to $\sqrt{\lambda_1}$, but now turns
into an increasing piecewise-linear function of $k$. When the
lateral load $f$ vanishes, the null solution is unique provided
that $\beta>-\beta_c(k)$, and
buckles when $\beta$ exceeds this critical value. In general, as well as in
the case $k=0$, the set of buckled solutions is finite, but for some special positive values
of $k$, called {\it resonant values}, infinitely many solutions
may occur (see Theorem \ref{TH-stat.solut}).

By paralleling the results for $k=0$, the null solution is
expected to be exponentially stable, when it is unique. Quite
surprisingly, it is not so. For large values of $k$, the energy
decays with a sub-exponential rate when
$-\bar\beta>\beta>-\beta_c$ (see Theorem \ref{exp-stab}). In
particular, for any fixed $k>\lambda_1$, the positive limiting
value $\bar\beta(k)$ is smaller than the critical value
$\beta_c(k)$, and the former overlaps the latter only if $0\leq
k\leq \lambda_1$. A picture of these functions, as $k$ runs the
positive axis, is given in Fig. \ref{fig}.

The plan of the paper is as follows. In Section 2 we discuss the
general functional framework of  (\ref{PROB}) and exact
solutions for the stationary postbuckling problem is presented for
all $k>0$, when $f=0$. After formulating an abstract version of the
dynamical problem, the existence of an absorbing set is addressed
in Section 3. Some preliminary estimates and the exponential stability result are established in Section 4. The
main result concerning the existence of a regular global attractor
is stated in Section 5, where the asymptotic smoothing property of
the semigroup generated by the abstract problem is proved via a
suitable decomposition first proposed in \cite{GPV}.

Although we assume here that both ends of the beam are hinged, different boundary conditions for $u$ are also physically significant, such as when both ends are clamped, or one end hinged and the other one clamped.  On the contrary, the so-called cantilever boundary condition (one end clamped and the other one free) is not covered because it is pointedly inconsistent with the extensibility assumption of the model.
Nevertheless, the hinged case we consider here is very special. Indeed, other boundary conditions lead to a completely different analysis that must take into consideration very special estimates for the complementary traces on the boundary, and only weaker forms of the regularity properties of solutions remain valid (see, for instance, \cite{LT}).

It is worth noting that several papers (see, for instance,
\cite{CC,CCK,Ma}) are devoted to approximations as well as
numerical simulations in the modelling of the deformations of
extensible beams on elastic supports. In this connection, our
paper which exhibits exact solutions is of interest in order to
fit computer applications.
In particular, our model can be useful for engineering
applications involving simply supported bridges subjected to
moving vertical loads. For instance, it may be adapted to the
study of lively footbridges \cite{venuti1,venuti2}.

Finally, we remark that our analysis is carried over an abstract version of the original problem which is independent of the space dimension,
so that it could be extended to scrutinize shear deformations in plate models.
The techniques of this work apply to plate models as well, without substantial changes.

In addition, our strategy can be generalized to the investigation
of non-linear dissipative models which describe the
vibrations of extensible viscoelastic beams where the dissipative
term derives from the internal viscoelastic dissipation (memory).
Moreover, we are going to scrutinize the longterm damped dynamics of extensible elastic bridges suspended by flexible and elastic
cables. In this model, the term $-ku$ is replaced by $-ku^{+}$ and
it represents a restoring force due to the cables, which is
different from zero only when they are being stretched.
\section{Stationary solutions}
Our aim is to analyze the multiplicity of solutions to the boundary value problem \eqref{STATIC}. Letting $L^2(0,1)$ the Hilbert space of square summable functions on $(0,1)$, the domain of the differential operator $\partial_{xxxx}$ appearing in \eqref{STATIC} is
$$
\D(\partial_{xxxx})=\{w\in H^4(0,1) : w(0)=w(1)=w''(0)=w''(1)=0\}.
$$
This operator is strictly positive selfadjoint with compact inverse, and its discrete spectrum is given by $\lambda_n=n^4\pi^4$, $n\in\N$.
Thus, ${\lambda_1=\pi^4}$ is the smallest eigenvalue. Besides, the following relation holds true
\begin{equation}
\label{DELTA}
(\partial_{xxxx})^{1/2}=-\partial_{xx}\,,\qquad\D(-\partial_{xx})=H^2(0,1)\cap H^1_0(0,1).
\end{equation}
For every $k>0$, let
$$
\mu_n(k)=\frac{k}{n^2\pi^2}+ n^2\,\pi^2\,,\qquad \beta_c(k)=
\min_{n\in\N}\mu_n(k)\ .
$$
Assuming that $n_k\in\N$ be such that
$\displaystyle\mu_{n_k}=\min_{n\in\N}\mu_n(k)$, then it satisfies
$$
(n_k-1)^2n_k^2\leq \frac k {\pi^4}< n_k^2(n_k+1)^2\,.
$$
As a consequence, $\beta_c(k)$ is a piecewise-linear function of
$k$ (see Fig. \ref{fig}). The set
$${\mathcal R}=\{ i^2j^2\pi^4:i,j\in\N, i<j\}$$
is referred as the {\it resonant set}: when $k\in{\mathcal R}$
there exists at least a value   $\mu_j(k)$ which is not simple.
Indeed,
 $\mu_i=\mu_j$, $i\neq j$, provided that $k= i^2j^2\pi^4$ (resonant values). In the sequel, let $\mu_m(k)$ be the smallest value of $\{\mu_n\}_{n\in\N}$ which is not simple. Of course, the $\mu_n(k)$ are all simple and increasingly ordered with $n$ whenever $k<4\pi^4$.  Given $k>0$, for later convenience let $n_\star$ be the integer-valued function given by
 \begin{equation}\nonumber
n_{\star }(\beta)=|{\mathcal N}_\beta|\,, \qquad{\mathcal
N}_\beta= \{n\,\in \mathbb{N}:\beta +\mu_n(k)<0\},
\end{equation}
where $|{\mathcal N}|$ stands for the cardinality of the set
${\mathcal N}$.

In the homogeneous case, we are able to establish the exact number
of stationary solutions of (\ref{STATIC}) and their explicit form.
In particular, we will show that there is always at least one
solution, and at most a finite number of solutions, whenever the
values of $\mu_n(k)$ not exceeding $-\beta$ are simple.

\begin{theorem}
\label{TH-stat.solut}
 If $\beta \geq - \beta_{c}(k)$, then for every $k>0$ system
 (\ref{STATIC})
 with $f=0$ has the null solution, corresponding to the straight equilibrium position. Otherwise:
\begin{itemize}
\item if $k\in{\mathcal R}$ and $\beta < -\mu_m(k)$, the smallest
non simple eigenvalue, there are infinitely many solutions; \item
whether $k\in{\mathcal R}$ and $-\mu_m(k)\leq \beta < -
\beta_{c}(k)$, or $k\not\in{\mathcal R}$ and $\beta < -
\beta_{c}(k)$, then besides the null solution there are also
$2n_\star (\beta)$ buckled solutions, namely
\begin{equation}
\label{soluz.tutte} u_{n}^{\pm }(x)\,=\,A_{n}^{\pm }\,\sin (n\pi
x)\,,\quad\,n=1,2,\ldots ,n_{\star }
\end{equation}
with
\begin{equation}
\label{staticresponse} A_{n}^{\pm }=\pm \,\frac{1}{n\,\pi
}\sqrt{-\,2\,\left[ \beta + \mu_n(k)\right]}.
\end{equation}
\end{itemize}
\end{theorem}
\begin{proof}
Clearly, $u=0$ is a solution to (\ref{STATIC}) in the homogeneous
case for all $k$ and $\beta$. To find a  nontrivial solution $u$,
we put $h\,=\,\beta + \int\limits_{0}^{1}\left| u^{\prime
}(\xi)\right| ^{2}d\xi$\,, so that $u$ solves the differential
equation
\begin{equation*}
u^{\prime \prime \prime \prime }\,-\,h\,u^{\prime \prime}+ku=0 \ ,
\quad \quad h \in \mathbb{R}\,,\, k>0\,.
\end{equation*}
Letting $\lambda^{2}\,=\chi $, the characteristic equation
$$
\lambda ^{4}\,-\,h\,\,\lambda ^{2}+k=0
$$
admits solutions in the form
\begin{equation}
\label{SS2} \chi_{1,2}=\frac{h\pm \sqrt{h^{2}-4k}}{2} \ .
\end{equation}
As a consequence, taking into account the hinged boundary
conditions, we obtain
\begin{itemize}
\item if $h\geq 2 \sqrt{k}$, then $ \chi_{1,2} \in \mathbb{R}^+$ and
all corresponding values of $\lambda$ are real, hence
$u \equiv 0$; \\

\item if $h \leq -2\sqrt{k}$, then $\chi_{1,2} \in \mathbb{R}^-$
and $\lambda_{1,2}=\pm \sqrt{| \chi _{1}| }i$,
$\lambda_{3,4}=\pm\sqrt{| \chi _{2}| }i$; hence
$$
u=2\,i\,a\,\sin \omega _{1}\,x+2\,i\,b\,\sin \omega _{2}\,x,
$$
\end{itemize}

\noindent where $a$ and $b$ are suitable constants, while
\begin{equation}
\label{OMEGA} \omega _{1}=\sqrt{\left| \chi _{1}\right| }=n\pi \,,
\qquad
 \omega _{2}=\sqrt{\left| \chi _{2}\right| }=\ell\pi\,,\qquad n,\ell\in\N .
\end{equation}
Now, letting $2\,a\,=\,-i\tilde{A}$ and $2\,b\,=\,-i\tilde{B}\,$,
we can write the solution into the form
\begin{equation}
\label{SOLUZIO} u=\tilde{A}\,\sin n\pi\,x+\tilde{B}\,\sin
\ell\pi\,x .
\end{equation}
Moreover, from (\ref{SS2}), (\ref{OMEGA}) we obtain
\begin{equation}
\label{ValoreH1} \frac1{2} \big(-h-\sqrt{h^{2}-4k}\,\big)=n^{2}\pi
^{2} \Rightarrow h=-\mu_n(k)\,=\,\beta + \int\limits_{0}^{1}\left|
u^{\prime }(\xi)\right| ^{2}d\xi
\end{equation}
\begin{equation}
\label{ValoreH2} \frac1{2}
\big(-h+\sqrt{h^{2}-4k}\,\big)=\ell^{2}\pi ^{2} \Rightarrow
h=-\mu_{\ell}(k)=\,\beta + \int\limits_{0}^{1}\left| u^{\prime
}(\xi)\right| ^{2}d\xi
\end{equation}
from which it follows $\mu_{\ell}(k)=\mu_n(k)$. In order to
represent $\int\limits_{0}^{1}\left| u^{\prime }(\xi)\right|
^{2}d\xi $ in explicit form, we are lead to consider two
occurrences.
\begin{itemize}
\item
 Let $k\not\in{\mathcal R}$. In this case, from (\ref{ValoreH1}) and (\ref{ValoreH2})  it follows $\ell=n$ and, by (\ref{SOLUZIO}),
the following equality holds
$$
\int\limits_{0}^{1}\left|
u^{\prime }(\xi)\right| ^{2}d\xi =\frac{1}{2}\left(
\tilde{A}+\tilde{B}\right) ^{2}n^{2}\pi ^{2}=-\beta-\mu_n(k),
$$
\noindent so that, letting $A=\tilde{A}+\tilde{B}$,
(\ref{staticresponse}) can be  easily obtained. Of course, such
nontrivial solutions exist if and only if
$\displaystyle-\beta>\beta_c(k)= \min_{n\in\N}\mu_n(k) $. \item
 Let $k\in{\mathcal R}$. Then,
 $\mu_{\ell}(k)=\mu_n(k)$ for some $\ell< n$. In this case, from (\ref{SOLUZIO}) we obtain the equality
$$
\int\limits_{0}^{1}\left|
u^{\prime }(\xi)\right| ^{2}d\xi =\frac{1}{2}
\tilde{A}^{2}n^{2}\pi
^{2}+\frac{1}{2}\tilde{B}^{2}\ell^{2}\pi^{2}, \quad \ell\neq n,
$$
and (\ref{ValoreH1})-(\ref{ValoreH2}) cannot uniquely determine
the values of $\tilde{A}$ and  $\tilde{B}$. Accordingly,
(\ref{SOLUZIO}) represent infinitely many solutions provided that
$-\beta>\mu_m(k)$.

  \end{itemize}
\end{proof}

\begin{remark}
Assuming $k=0$ we recover the results of \cite{NE} and
\cite{CZGP}.
\end{remark}

\begin{figure}[h]
\begin{center}
 \psfrag{O}{$O$}
 \psfrag{k}{$k$}
 \psfrag{l}{\tiny{$-\pi^2$}}
\psfrag{m}{\scriptsize{$-\disp\frac{k}{\pi^2}-\pi^2$}}
 \psfrag{A}{$A$}
 \psfrag{A1+}{$A_1^+$}
 \psfrag{A2+}{$A_2^+$}
 \psfrag{A1-}{$A_1^-$}
 \psfrag{A2-}{$A_2^-$}
 \psfrag{u1+}{$u_1^+$}
 \psfrag{u1-}{$u_1^-$}
 \psfrag{u2-}{$u_2^-$}
 \psfrag{u2+}{$u_2^+$}
  \psfrag{b}{$\beta$}
   \psfrag{n}{\tiny{$-\disp\frac{k}{4\pi^2}-4\pi^2$}}
   \psfrag{p}{\tiny{$-\disp\frac{k}{9\pi^2}-9\pi^2$}}
\centerline{\includegraphics[width=12cm]{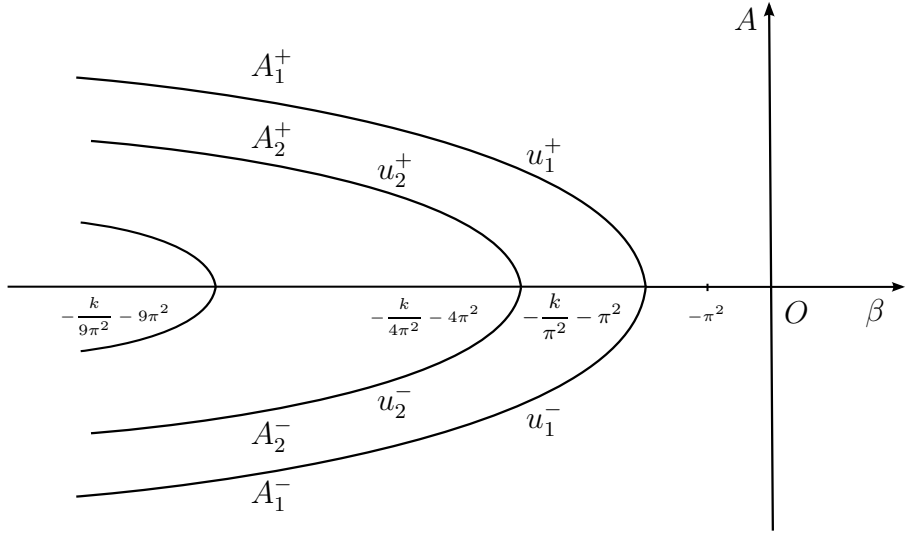}}
\caption{ A sketch of the nonlinear static response of the beam
when $k<4\pi^4$.} \label{figureresponse}
\end{center}
\end{figure}
\noindent When  $k\not\in{\mathcal R}$ the set of all stationary
states is finite and will be denoted by $\mathcal{S}_0$. Depending
on the values of $k$ and $\beta$, the pairs of solutions branch
from the unbuckled state $A^\pm_n\,=\,0$ at the critical value
$\beta =  - \beta_{c}(k)$, i.e., the beam can buckle in either the
positive or negative directions of the transverse displacement.
These branches exist for all $\beta <  - \beta_{c}(k)$ and
$A^\pm_n$ are monotone increasing functions of $|\beta|$. For each
$n$, (\ref{staticresponse}) admits real (buckled) solutions
$A_{n}^{\pm }$ if and only if $\beta < -\mu_n$. When $k<4\pi^4$,
for any $\beta$ in the interval
$$ - \frac{k}{(n+1)^2\,\pi^2}-(n+1)^2\,\pi^2<\beta < - \frac{k}{n^2\,\pi^2}-n^2\,\pi^2 \ ,$$
the set $\mathcal{S}_0$ contains exactly $2n_{\star}+1$ stationary
points: the null solution and the pairs of solutions represented
by (\ref{soluz.tutte}). These properties are sketched in
Fig.\ref{figureresponse} (see also \cite{NE}).

\section{The Absorbing Set}

\noindent We will consider an abstract version of problem~(\ref{PROB}). To this aim, let $H_0$ 
be a real Hilbert space, whose inner product and norm are denoted
by $\l\cdot,\cdot\r$ and $\|\cdot\|$, respectively. Let
$A:\D(A)\Subset H_0\to H_0$ be a strictly positive selfadjoint
operator. We denote by $\lambda_1>0$ the first eigenvalue of $A$.
For $\ell\in\R$, we introduce the scale of Hilbert spaces
$$
H_\ell=\D(A^{\ell/4}),\qquad \l u,v\r_\ell=\l
A^{\ell/4}u,A^{\ell/4}v\r,\qquad \|u\|_\ell=\|A^{\ell/4}u\|.
$$
In particular, $H_{\ell+1}\Subset H_\ell$ and the following scale
of Poincar\'e inequalities holds
\begin{equation}
\label{POINCARE} \sqrt{\lambda_1}\,\|u\|_\ell^2\leq
\|u\|_{\ell+1}^2.
\end{equation}
Finally, we define the product Hilbert spaces
$$
\H_\ell=H_{\ell+2}\times H_\ell.
$$

\subsection*{The abstract problem}
For $\beta\in\R$ and $f\in H_{0}$, we investigate the following
evolution equation on $\H_0$
\begin{equation}
\label{BASE} \ptt u+ Au+ \big(\beta+\|u\|^2_1\big)A^{1/2}u+\delta
\partial_t u + k\,u = f,
\end{equation}
in the unknowns $u(t):[0,\infty)\to H_2$ and $\pt
u(t):[0,\infty)\to H_0$, with initial conditions
$$
(u(0),\partial_t u(0))=(u_0,u_1)=z_0\in\H_0.
$$

\noindent Problem~(\ref{PROB}) is just a particular case of the
abstract system (\ref{BASE}), obtained by setting $H_0=L^2(0,1)$
and $A$ the realization of $\partial_{xxxx}$ in $H_0$.

Eq.~(\ref{BASE}) generates a strongly continuous semigroup (or
dynamical system) $S(t)$ on $\H_0$: for any initial data
$z_0\in\H_0$, $S(t)z_0$ is the unique weak solution to
(\ref{BASE}), with related (twice the) energy given by
$$\E(t)=\|S(t)z_0\|^2_{\H_0}=\,\|u(t)\|^2_2+\|\pt u(t)\|^2.$$
Besides, $S(t)$ continuously depends on the initial data. We omit
the proof of these facts, which can be demonstrated  by means of a
Galerkin procedure (e.g., following the lines of \cite{IB}). The
crucial point in applying this technique is to have uniform energy
estimates on any finite time-interval. As it will be apparent,
these estimates are easily implied by the uniform inequalities on
$(0,+\infty)$ established in the subsequent sessions.

Now, we prove the existence of the so called absorbing set for the
flow generated by problem (\ref{BASE}), that is, a bounded set
into which every orbits eventually enters. Such a set is defined
as follows:

\begin{definition} Let $B(0,R)$ be the open ball with center $0$ and radius $R>0$ in $\mathcal{H}_0$. A bounded set
$\mathcal{B}_{\H_0} \subset \mathcal{H}_0$ is called an {\em
absorbing set } for the semigroup  $S(t)$ if, for any $R>0$ and
any initial value $z_0 \in B(0,R)$, there exists $t_0(R)>0$ such
that
$$z(t) \in \mathcal{B}_{\H_0}\quad \forall \,\,t\geq t_0,$$
where $z(t)\,=\,S(t)z_0$ is the solution starting from $z_0$.
\end{definition}

The main result of this section follows from two lemmas, involving
the functionals $\L(t)$ and $\Phi(t)$ defined as
\begin{equation}
\label{L_FUNZIONALE} \L(t)\,=\,\E(t)+\frac{1}{2}\left( \beta
+\left\| u(t)\right\| _{1}^{2}\right) ^{2}+k\,\left\| u(t)\right\|
^{2}\geq\E(t)\geq 0,
\end{equation}
\begin{equation}
\label{PHI_FUNZIONALE} \Phi (t)\,=\,\L(t)+\varepsilon \left\langle
\partial_t u\,,\,u\right\rangle .
\end{equation}
\begin {lemma}
\label {lemma nuovo} For all $t>0$ and $z_0\in \H_0$ with
$\|z_0\|_{\H_0}\leq R$, there exists a positive constant $C$
(depending on $\left\| f\right\|$ and $R$) such that
\begin{equation}
\label{energy} \E(t)\leq C.
\end{equation}
\end{lemma}
\begin{proof}
If we consider the functional
$$
\F(t)\,=\,\L(t) - 2\left\langle f\,,\,u(t)\right\rangle ,
$$
from the energy identity
\begin{equation} \nonumber
\begin{split}
\frac{d \E}{dt} =-2\left\langle \left( \beta +\left\| u\right\|
_{1}^{2}\right) A^{\frac{1}{2}}u+ku+\delta \,\partial_t
u-f\,,\,\partial_t u\right\rangle
\end{split}
\end{equation}
 it easily follows the decreasing monotonicity of $\F$
\begin{equation} \nonumber
\begin{split}
\frac{d \F}{dt} &=\frac{d \E}{dt}+2\left( \beta +\left\| u\right\|
_{1}^{2}\right) \left\langle
A^{\frac{1}{4}}u\,,A^{\frac{1}{4}}\partial_t u\right\rangle
+2k\left\langle u,\partial_t u\right\rangle - 2\left\langle
f\,,\,\partial_t u\right\rangle\\ &= -2\delta \left\| \partial_t
u\right\| ^{2}\leq 0
\end{split}
\end{equation}
and then
\begin{equation} \nonumber
\begin{split}
\F(t) \leq \F(0)\leq C_0(R, \|f\|).
\end{split}
\end{equation}
Taking into account that
$$\left\| u(t)\right\| ^{2} \leq \frac{1}{\lambda_1}\left\| u(t)\right\|_{2} ^{2} \leq \frac{1}{ \lambda_1}\E(t) = C_1\E(t) ,$$
we obtain the estimate
$$\F(t) \geq \E(t) - 2\left\langle f\,,\,u(t)\right\rangle \geq \E(t) - \frac{1}{\varepsilon} \left\| f\right\| ^{2} - \varepsilon\left\| u(t)\right\| ^{2} \geq
(1 - \varepsilon C_1)\E(t) - \frac{1}{\varepsilon} \left\|
f\right\| ^{2}.$$ Finally, fixing $\varepsilon < \frac{1}{C_1}$,
we have
$$\E(t)\leq \frac{1}{1 - \varepsilon C_1}(\F(0) +  \frac{1}{\varepsilon} \left\| f\right\| ^{2}) \leq \frac{1}{1 - \varepsilon C_1}(C_0(R, \|f\|) +  \frac{1}{\varepsilon} \left\| f\right\| ^{2}) = C.$$
\end{proof}

\begin{lemma}
\label{lemma_34} For any given $z \in \mathcal{H}_0$ and for any
$t>0$ and $\beta \in \mathbb{R}$, when $\varepsilon$ is small
enough there exist three positive constants, $m_0 $, $m_1 $  and
$m_2$, independent of $t$ such that
\begin{equation}
m_{0}\,\mathcal{E}(t)\leq \Phi(t)\leq
m_{1}\,\mathcal{E}(t)\,+\,m_2.
\end{equation}
\end{lemma}
\begin{proof}
In order to prove the lower inequality we must observe that, by
Young inequality
$$
\left| \left\langle \partial_t u\,,\,u\right\rangle \right| \geq
-\frac{1}{2}\,\left\| \partial_t u\right\| ^{2}-\frac{1}{2}\left\|
u\right\| ^{2} \ ;
$$
hence, we obtain
$$
\Phi (t)\, \geq
\,\|u(t)\|^2_2+\left(1-\frac{\varepsilon}{2}\right)\|\pt
u(t)\|^2+\frac{1}{2}\left( \beta +\left\| u(t)\right\|
_{1}^{2}\right)
^{2}+\left(k\,-\frac{\varepsilon}{2}\right)\,\left\| u(t)\right\|
^{2} \ .
$$
If we choose $\varepsilon$ small enough to satisfy $\varepsilon<2$
and $\varepsilon<2k$, then we have
\begin{equation} \label{stima}
\Phi (t)\, \geq m_0 \, \L(t) \geq m_0 \, \E(t),
\end{equation}
where $m_0\,=\,\min\{1- \frac{\varepsilon}{2}, 1-
\frac{\varepsilon}{2k}\}$.

The upper inequality can be obtained using the definition of
$\Phi$ and applying the estimate
\begin{equation} \label{aggiunta}
\left| \left\langle \partial_t u\,,\,u\right\rangle \right| \leq
\frac{1}{2}\,\left\| \partial_t u\right\|
^{2}+\frac{1}{2\lambda_1}\left\| u\right\|_2^{2} \ .
\end{equation}
First, we can write
\begin{equation} \nonumber
\begin{split}
\Phi (t) \leq
 \left[ 1+\frac{1}{ \,\lambda _{1}}\left(
k+\frac{\varepsilon }{2}\right) \right] \left\| u(t)\right\|
_{2}^{2}+\left( \frac{\varepsilon }{2}+1\right) \left\|
\partial_t u(t)\right\| ^{2}+\frac{1}{2}\left( \beta +\left\| u(t)\right\|
_{1}^{2}\right) ^{2}.
\end{split}
\end{equation}
Then, by (\ref{POINCARE}) and Lemma \ref{lemma nuovo} we infer
\begin{equation} \label{stima2}
\left( \beta +\left\|u\right\|_1 ^{2}\right)  \leq \left| \beta
\right| +\frac{1}{ \sqrt{\lambda _{1}} }\,C = \bar {C},
\end{equation}
so that we finally obtain
\begin{equation} \nonumber
\begin{split}
\Phi (t) &\leq \left[ 2+\frac{1}{ \,\lambda _{1}}\left(
k+\frac{\varepsilon }{2}\right) +\frac{\varepsilon }{2}\right]
\E(t) +\frac{1}{2}\bar {C}^2 \, = \, m_1\,\E(t) + m_2 \ .
\end{split}
\end{equation}
\end{proof}

\begin{theorem}
\label{th_ass.set} For any $\beta \in \mathbb{R}$, there exists an
absorbing set in $\mathcal{H}_0$ for the dynamical system
$\left(S(t), \mathcal{H}_0 \right)$.
\end{theorem}
\begin{proof}

By virtue of (\ref{L_FUNZIONALE}) we have
\begin{equation} \nonumber
\begin{split}
\frac{d \L}{dt} &=\frac{d \E}{dt}+2\left( \beta +\left\| u\right\|
_{1}^{2}\right) \left\langle
A^{\frac{1}{4}}u\,,A^{\frac{1}{4}}\partial_t u\right\rangle
+2k\left\langle u,\partial_t u\right\rangle = %
\\
&= -2\delta \left\| \partial_t u\right\| ^{2}+2\left\langle
f\,,\,\partial_t u\right\rangle.
\end{split}
\end{equation}
Moreover, by (\ref{PHI_FUNZIONALE})
\begin{equation} \nonumber
\begin{split}
\frac{d \Phi }{dt} &=\frac{d \L}{dt}+\varepsilon \left\langle
u\,,\partial_{tt} u\right\rangle +\varepsilon \left\| \partial_t
u\right\| ^{2}=-\left( 2\delta -\varepsilon \right) \left\|
\partial_t u\right\| ^{2}+2\left\langle f\,,\,\partial_t
u\right\rangle +
\\
& +\varepsilon \left[ - \left\| u\right\| _{2}^{2}-\delta
\,\left\langle \partial_t u,\,u\right\rangle -k\left\| u\right\|
^{2}-\beta \left\| u\right\| _{1}^{2}-\left\| u\right\|
_{1}^{4}+\left\langle f\,,\,u\right\rangle \right].
\end{split}
\end{equation}
A straightforward computation leads to the identity
\begin{equation}
 \label{deri-Phi}
\begin{split}
&\frac{d \Phi}{dt} +\varepsilon \Phi + 2\left( \delta -\varepsilon
\right) \left\| \partial_t u\right\| ^{2}+\frac{\varepsilon
}{2}\left\| u\right\| _{1}^{4}=
\\
&=2\left\langle f\,,\,\partial_t u\right\rangle +\varepsilon
\left[ \left\langle f\,,\,u\right\rangle -\left( \delta
-\varepsilon \right) \,\left\langle \partial_t u,\,u\right\rangle
\right] +\frac{\varepsilon }{2}\beta ^{2}\ .
\end{split}
\end{equation}
In the following we estimate the terms in the rhs of the previous
equality. By means of the H$\ddot{o}$lder and Young's inequalities
and \eqref{aggiunta}, we have
%
%
%
%
$$
\frac{d \Phi}{dt} +\varepsilon \Phi +\frac{1}{2}\left( 3\delta
-5\varepsilon \right)\left\|
\partial_t u\right\| ^{2}\leq
$$

$$
\frac{1}{2}\varepsilon \beta ^{2}+\left( \frac{1}{\varepsilon
}+\frac{1 }{2}\right) \left\| f\right\| ^{2}+
\frac{\varepsilon^2(\delta-\varepsilon+1)}{2\,\lambda_1} \left\|
u\right\| _{2}^{2}.
$$
\newline
Now, choosing $\varepsilon < \frac{3}{5}\,\delta$ and $\varepsilon
< 1 + \delta$ we find
$$
\frac{d \Phi}{dt} +\varepsilon \Phi \leq \frac{1}{2}\varepsilon
\beta ^{2}+\left( \frac{1}{\varepsilon }+\frac{1}{2}\right)
\left\| f\right\| ^{2}+m \, \varepsilon ^{2}\E,
$$
where
$$
 m\, =\frac{1}{2\lambda _{1}} \left( 1 + \delta
-\varepsilon \right)>0.
$$
Finally, from Lemma \ref{lemma_34} we obtain
$$
\frac{d \Phi}{dt} +\varpi \Phi \leq \frac{1}{2}\varepsilon \beta
^{2}+\left( \frac{1}{\varepsilon }+\frac{1}{2}\right) \left\|
f\right\| ^{2},
$$
where $\varpi =\varepsilon \left(1- \frac m{m_0}\varepsilon
\right) $ is positive provided that $\varepsilon<\frac{m_0} m$.
Since $\Phi$ is positive also when $\beta<0$ provided that
$\varepsilon$ is chosen small enough, using Gronwall lemma it
follows that
\begin{equation} \label{exp1}
\Phi \left( t\right) \leq \Phi \left( 0\right) \,e^{-\varpi
\,t}+\frac{1}{2}\varepsilon \beta ^{2}+\left( \frac{1}{\varepsilon
}+\frac{1 }{2}\right) \left\| f\right\| ^{2},
\end{equation}
and accordingly
\begin{equation} \label{exp2}
\E(t) \leq \frac 1 {m_0}\Phi(t) \leq \Gamma_0(R)
e^{-\varpi\,t}+\Gamma_1(\beta,\,\|f\|),
\end{equation}
where
$$\Gamma_0(R)\,=\,\frac{\Phi(0)}{m_0} \qquad {\rm{and}}
\qquad \Gamma_1(\beta,
\|f\|\,)\,=\,\frac{1}{m_0}\left[\frac{1}{2}\varepsilon \beta
^{2}+\left( \frac{1}{\varepsilon }+\frac{1 }{2}\right) \left\|
f\right\| ^{2}\right].$$
As a consequence, every ball $B\left( 0,\bar{R}\right) $ in $\H_0$
with radius $\bar{R}>1+\Gamma_1(\beta,\,\|f\|)$ can be chosen as
an absorbing set in that it verifies the following statement:
\newline
\it{for all $ z_{0}=(u_{0},\,u_{1}\,)\in B(0,R)$, there exists
$t_{0}(R)=\frac{1}{\varpi }\log \Gamma_0(R)$ such that for any
$t>t_{0},\,\,\,z(t)\,\in B(0,\bar{R}).$}\\
\end{proof}

\section{Exponential Stability}
A direct proof of the exponential decay of the energy seems out of
reach, so we exploit the equivalence between the energy
$\mathcal{E}$ and the functional
$$\bar{\Phi}= \Phi-\frac{1}{2}\beta^2$$
which can be proved to be exponentially
stable. The positivity of such a functional will be obtained as a
direct corollary of the next Lemma \ref{lemma}.

Recalling Theorem \ref{TH-stat.solut}, the set $\mathcal{S}_0$ of
stationary solutions reduces to a singleton when
\begin{equation}
\label{beta} \beta \geq -\beta_c(k)= -\min_{n\in\N}\mu_n(k),\qquad
\mu_n(k)= \sqrt{\lambda_n}\left[1+\frac{k}{\lambda_n}\right], \ \
\lambda_n= n^4\pi^4.
\end{equation}
It is worth noting that $\beta_c(k)$ is a piecewise-linear
function of $k$, in that $\beta_c(k)=\mu_1(k)$ when
$0<k<\sqrt{\lambda_{1}} \sqrt{\lambda_{2}}=4\pi^4$, and in general
$$\beta_c(k)=\mu_n(k)\quad  \hbox{when}\quad  \sqrt{\lambda_{n-1}}\sqrt{\lambda_{n}}<k< \sqrt{\lambda_{n}}\sqrt{\lambda_{n+1}}.$$
Unlike the case $k=0$, the energy $\E(t)$ does not decay
exponentially in the whole domain of the $(\beta,k)$ plain where
(\ref{beta}) is satisfied, but in a region which is strictly
included in it.

For futher purposes, let
\begin{equation} \nonumber
\bar{\beta}(k)= \left\{
\begin{array}{ll}
\beta_{c}(k) \quad \quad \quad \qquad 0<k\leq \lambda_1,
\\[1em]
2\sqrt{k} \quad \quad \quad \qquad k>\lambda_1.
\end{array}
\right.
\end{equation}
A picture of this function is given in Fig. \ref{fig}.

\begin{lemma} \label{lemma}
Let $\beta \in \mathbb{R}$, $k>0$ and
$$
Lu\,=Au\,+\,\beta A^{\frac{1}{2}}u\,+k\,u \, .
$$
There exists a real function $\nu\,=\,\nu(\beta,k)$ such that
$$\left\langle Lu\,,\,u\right\rangle \geq \nu \left\| u\right\| _{2}^{2},$$
where $\nu(\beta,k)>0$ if and only if $\beta >-\bar{\beta}(k)$.
\end{lemma}
\begin{proof}
Taking into account the inner product
$$\left\langle Lu\,,\,u\right\rangle = \left\| u\right\|
_{2}^{2}+\,\beta \left\| u\right\| _{1}^{2}+k\left\| u\right\|
^{2},$$
we put
$$
X=\left\| u\right\| _{2}\,;\,Y=\left\| u\right\|
_{1}\,;\,Z=\left\| u\right\|
$$
so that
$$\left\langle Lu\,,\,u\right\rangle = I(X,Y,Z)= \,X^{2}\,+\,\beta \,Y^{2}\,+k\,Z^{2}.$$
If $\beta \geq 0$, the desired inequality trivially holds true by choosing $\nu\,=\,1$.
\par
Letting $\beta<0$, by means of  the interpolation inequality $\left\|
u\right\| _{1}^{2}\leq \left\| u\right\| _{2}\left\| u\right\|$,
we obtain
$$I(X,Y,Z)\geq  \,X^{2}\,+\,\beta XZ+k\,Z^{2}=J(X,Z).$$
Hence the thesis can be rewritten as follows: find $\nu>0$ such that
$$J(X,Z)\mid _{D_{0}}\geq \nu \,X^{2},$$
where $D_{0}=\left\{ \left( X,Z\right) \,:\,X\geq 0,\,Z\geq
0,\,\,0\leq Z\leq X/{\sqrt{\lambda _{1}}}\right\}.$
\\[1em]
In order to prove this statement, we introduce the set
$$ M=\left\{ m\;\in
\mathbb{R}:\,0\leq m\leq \frac{1}{\sqrt{\lambda _{1} }}\right\},
$$ so that $D_{0}=\left\{ \left( X,Z\right) \,:Z= m X,\,X\geq 0
\,\, {\rm and} \,\, m \in M\right\}$ and the original problem
reads: find $\nu>0$ such that
$$\,J(X,Z)\mid _{Z=mX}\geq \nu _{m}\,X^{2}\,\,,\quad \forall m \in M$$
and
$$
\nu =\underset{m\in M}{\inf }\nu _{m}>0.
$$
\newline
We first observe that
$$\,J(X,Z)\mid _{Z=mX}=\left(1+\beta\,m+k\,m^{2}\right) X^{2}=\eta (m)X^{2}
$$
where $\eta (m)$ is a concave parabola. Hence, we have to find the
region in the $(\beta,k)$-plane where $\eta$ admits a strictly
positive minimum on $M$, that is to say $\eta _{\min }=\nu>0$. We
shall prove that $\nu>0$ if and only if  $\beta >
-\bar{\beta}(k)$. To this end, we split the proof into four steps.
\begin{itemize}

    \item Step 1. We consider the region $R_1=\left\{(\beta, k)\,:\,k > 0, \, -2\sqrt{k}<\beta <0\right\}$.\\
    Since the discriminant of the parabola is negative, $ \Delta _{\eta }=\beta ^{2}-4 k<0$,
    the value $\eta(m)$ is strictly positive for all $m$ in the closed interval $M$, and then $\nu>0$.\\

    \item Step 2. Let
    $R_2=\left\{(\beta, k): 0< k \leq \lambda_1, -\sqrt{\lambda_1}-{k}/{\sqrt{\lambda_1}}<\beta \leq -2 \sqrt{\,k}
    \right\}$.\\
    Observing that $\lambda_1<\sqrt{\lambda_1}\sqrt{\lambda_2}$, we infer that $\beta_c(k)=\mu_1=\sqrt{\lambda_1}+{k}/ {\sqrt{\lambda_1}}$.
    Now, since  $ \Delta _{\eta }=\beta ^{2}-4
    k \geq0$,  there exist  two solutions, $m_1\,,m_2\,\in \mathbb{R} $,
    of $\eta(m)=0$ so that $\eta$ changes sign on $\R$. In particular,
    it must be negative inside the open interval $(m_1\,,\,m_2)$
    because of $k>0$.
    On the other hand, $\eta$ is positive at the ends of $M$. Indeed, $\eta(0)\,=1$.
    In order to evaluate the sign of $\eta(1/{\sqrt{\lambda_1}})$
    we let $k\,=\,\rho\,\lambda_1$, with $0<\rho \leq 1$, and we obtain
    $$
    \eta\Big(\frac{1}{\sqrt{\lambda_1}}\Big)\,=\,\left( 1+\,\rho\right)\,
    +\,\frac{\beta}{\sqrt{\lambda_1}} \ .
    $$
    According to the definition of $R_2$
    \begin{equation}
    \label{rho}
    - \sqrt{\lambda _{1}}-\rho \,{\sqrt{\lambda _{1}}}<\beta <-2 \sqrt{\rho \lambda _{1}}
    \end{equation}
    therefore
    $$\frac{\beta }{\sqrt{\lambda _{1}}}+ \left( 1+\rho \right) >0,$$
    which implies
    $\eta \left(1/{\sqrt{\lambda _{1}}}\right) >0 .$
    Thus we infer that  either $(m_1\,,\,m_2)\subset M$, or $(m_1\,,\,m_2)$ is external to $M$.
    In the sequel we prove the latter occurrence by showing that the vertex of the parabola lies outside of $M$.
    Letting $m_*$ the abscissa of the vertex, it satisfies $\eta^\prime\mid _{m=m_*}= \beta + 2 k m_*\,=\,0$ so that
    $$m_*\,=\,-\frac{\beta}{2k}\,=\,-\frac{\beta}{2\,\rho\,\lambda_1}.$$
    So, in order to have $\eta(m) \mid_M\,\, >0$, it is enough to prove that
    $$-\frac{\beta }{2 \rho \lambda _{1}}=m_* \geq \frac{1}{\sqrt{\lambda_1}} \ ,$$
    which is trivially true because from (\ref{rho}) we have
        $$\beta \leq -2 \sqrt{\rho }\sqrt{\lambda _{1}}\leq -2 \rho \sqrt{\lambda _{1}}.$$

    \item Step 3. Let
     $R_3=\left\{(\beta\,,k)\,:\,k>\lambda_1,-\sqrt{\lambda_1}-{k}/{\sqrt{\lambda_1}}<\beta \leq -2 \sqrt{k}\right\}$.  When $\beta < -2 \sqrt{k}$,  as before we have  $ \Delta _{\eta }=\beta ^{2}-4
    k > 0$, so that $\eta$ is negative valued in the open interval $(m_1\,,\,m_2)$ delimited by solutions $m_1\,,m_2$ to the equation
    $\eta(m)=0$. Nevertheless,
    in this case  the vertex of the parabola lies inside $M$. Indeed, if we
     take $k\,=\,\sigma\,\lambda_1$, with $\sigma >
    1$, the abscissa of the vertex satisfies
    $$0<m_*=-\frac{\beta }{2 \sigma \lambda _{1}} <
    \frac{1}{\sqrt{\lambda_1}}\,,$$
    which holds true by virtue of the definition of  $R_3$, in that
        $$ -\beta<\left( 1+\sigma \right)  \sqrt{\lambda
    _{1}}<2 \sigma \sqrt{\lambda _{1}}.
    $$
    As a consequence, the minimum of $\eta$ on $M$ is $\eta(m_*)<0$. When $\beta = -2 \sqrt{k}$, we have $m_1=m_2=m_*$ and $\eta(m_*)=0$, so that the minimum of $\eta$ on $M$ vanishes. In both cases, $\eta _{\min }=\nu$ is not positive.\\

    \item Step 4. We consider the set $R_4=\left\{(\beta,k): k > 0 ,\beta <-\sqrt{\lambda_1}-{k}/{\sqrt{\lambda_1}}\right\}$.
          In this case $\eta(0)\,=1$, whilst
     $$\eta \left( \frac{1}{\sqrt{\lambda _{1}}}\right) \,=
     \frac{1}{\sqrt{\lambda _{1}}}\left( \sqrt{\lambda _{1}}+\beta +\frac{k}{\sqrt{\lambda _{1}}}\right) <0$$
     and the minimum of $\eta$ on $M$ cannot be positive.
\end{itemize}

\end{proof}

We are now in a position to prove the following
\begin{theorem} \label{exp-stab}
When $f\,=0$, the solutions to (\ref{PROB}) decay exponentially,
i.e.
$$
\E(t)\,\leq\,c_0\,\E(0)\,e^{- c t}
$$
with $c_0$ and $c$ suitable positive constants, provided that
$\beta\,>-\bar{\beta}(k)$.
\end{theorem}

\begin{proof}
Let $\bar{\Phi}$ the functional obtained from $\Phi$ by letting
$f=0$ and neglecting the term $\frac{1}{2}\beta ^{2}$, i.e.
\begin{equation} \nonumber
\bar{\Phi} (t)= \left\| u(t)\right\| _{2}^{2}+\left\| \partial_t
u(t)\right\| ^{2}+\beta \left\| u(t)\right\|
_{1}^{2}+\frac{1}{2}\left\| u(t)\right\| _{1}^{4}+k\left\|
u(t)\right\| ^{2}+\varepsilon \left\langle
\partial_t u(t)\,,\,u(t)\right\rangle.
\end{equation}
In view of applying Lemma \ref{lemma}, we remark that
$$
\bar{\Phi} =\left\langle Lu\,,\,u\right\rangle+\left\|
\partial_t u\right\| ^{2}+\frac{1}{2}\left\| u\right\| _{1}^{4}+\varepsilon
\left\langle u,\partial_t u\right\rangle.
$$
The first step is to prove the equivalence between $\mathcal{E}$
and $\bar{\Phi}$, that is
\begin{equation} \nonumber
c_{1}\,\,\,\mathcal{E}\leq \bar{\Phi}\leq c_{2}\,\,\mathcal{E}\ .
\end{equation}
We split the proof into two parts.
 \begin{itemize}
 \item {\it First step}: $c_{1}\,\,\mathcal{E} \leq \bar{\Phi} $\\
 \end{itemize}
By virtue of Lemma \ref{lemma} and (\ref{POINCARE}), the following
chain of inequalities holds provided that $\beta
>-\bar{\beta}(k)$, which ensures the positivity of $\nu$:
$$
\bar{\Phi}\geq  \left( \nu -\frac{\varepsilon }{2\,\lambda _{1}
}\right) \left\| u\right\| _{2}^{2}+\frac{1}{2}\left\| u\right\|
_{1}^{4}+\left( 1-\frac{\varepsilon }{2\,}\right) \left\|
\partial_t u\right\| ^{2} \geq
$$

$$
\geq  \left( \nu -\frac{\varepsilon }{2\,\lambda _{1} }\right)
\left\| u\right\| _{2}^{2}+\left( 1-\frac{\varepsilon
}{2\,}\right) \left\| \partial_t u\right\| ^{2}\geq \min
\left\{\nu -\frac{\varepsilon }{2\,\lambda _{1}
},1-\frac{\varepsilon }{2\,}\right\}\E\,.
$$
If we choose $\varepsilon < \min \left\{ 2 \nu \lambda_1\,,\,2
\right\}$ and we put $c_1=\min \left\{\nu -\frac{\varepsilon
}{2\,\lambda _{1} },1-\frac{\varepsilon }{2\,}\right\}$, it
follows
\begin{equation} \nonumber
c_1\,\,\mathcal{E}\,\leq\,\,\bar{\Phi}.
\end{equation}
\newline
 \begin{itemize}
 \item {\it Second step}: $\bar{\Phi}\leq c_{2}\,\,\mathcal{E}$\\
 \end{itemize}
Using the expression of $\bar{\Phi}$, Young inequality and
(\ref{POINCARE}), we can write the following chain of inequalities
$$
\bar{\Phi} 
%
\leq
%
%
 \left( 1+\frac{k}{ \lambda
_{1}}+\frac{1}{2 \,\lambda _{1}}\right) \left\| u\right\|
_{2}^{2}+\left( 1+\frac{\varepsilon ^{2}}{2}\right) \left\|
\partial_t u\right\| ^{2}+\beta \left\| u\right\|
_{1}^{2}+\frac{1}{2}\left\| u\right\| _{1}^{4}\leq
$$
$$
\leq \left( 2+\frac{k}{ \lambda _{1}}+\frac{1}{2 \,\lambda
_{1}}+\frac{\varepsilon ^{2}}{2}\right) \E+\beta \left\| u\right\|
_{1}^{2}+\frac{1}{2}\left\| u\right\| _{1}^{4} .
$$
In particular, from (\ref{stima2})
 we obtain
$$
\bar{\Phi}\leq \left( 2+\frac{k}{ \lambda _{1}}+\frac{1}{2
\,\lambda _{1}}+\frac{\varepsilon ^{2}}{2}\right)
\E+\,\,\bar{C}\,\left\| u\right\| _{1}^{2}\leq
$$

$$
\leq \left( 2+\frac{k}{ \lambda _{1}}+\frac{1}{2 \,\lambda
_{1}}+\frac{\varepsilon ^{2}}{2}+\frac{\bar{C}}{ \sqrt{\,\lambda
_{1}}} \,\right) \E  = c_2\,\,\mathcal{E}.
$$
The last step is to prove the exponential decay of $\bar{\Phi}$.
To this aim, remembering the hypothesis $f=0$, we can write the
identity (\ref{deri-Phi}) in the following way:
\begin{equation} \nonumber
\frac{d\bar{\Phi}}{dt}+\varepsilon \bar{\Phi}+2\left( \delta
-\varepsilon \right) \left\| \partial_t u\right\|
^{2}+\frac{\varepsilon }{2}\left\| u\right\| _{1}^{4}+\varepsilon
\left( \delta -\varepsilon \right) \left\langle \partial_t
u\,,\,u\right\rangle =0.
\end{equation}
Exploiting the Young inequality we obtain
\begin{equation} \nonumber
\frac{d\bar{\Phi}}{dt}+\varepsilon \bar{\Phi}+ \frac{3}{2} \left(
\delta -\varepsilon \right) \left\| \partial_t u\right\| ^{2}\leq
\frac{\varepsilon^2\left( \delta -\varepsilon
\right)}{2\lambda_{1} }\,\left\| u\right\| ^{2}_{2} \
\end{equation}
and choosing $\varepsilon<\delta$ it follows
\begin{equation} \nonumber
\frac{d\bar{\Phi}}{dt}+\varepsilon \bar{\Phi}\leq
\frac{\varepsilon^2\,\left( \delta -\varepsilon
\right)}{2\lambda_{1} } \,\E\leq \frac{\varepsilon^2\left( \delta
-\varepsilon \right)}{2\,c_{1}\lambda_{1} }\,\bar{\Phi}
\end{equation}
and finally, if $\varepsilon$ is small enough, we have
\begin{equation} \label{stima-exp1}
\frac{d\bar{\Phi}}{dt}+ c\,\bar{\Phi}\leq 0\ ,
\end{equation}
where $c = \varepsilon \left[ 1-\frac{\varepsilon \left( \delta
-\varepsilon \right) }{2\,c_{1}\lambda_{1} }\right]>0$. Equation
(\ref{stima-exp1}) implies
\\
$$c_{1}\E(t)\leq\bar{\Phi}\left( t\right) \leq \bar{\Phi}\left( 0\right)
\,e^{- \,c\,t}\leq c_{2}\,\E\left( 0\right)\,e^{- \,c\,t}.$$ The
thesis follows by letting $c_{0}=\frac{c_{2}}{c_{1}}$\ .
\end{proof}

\begin{remark}
We stress that Theorem \ref{exp-stab} holds even if $k=0$. In this
case however we have $\bar{\beta}(0)= - {\beta_{c}}(0)= -
\sqrt{\lambda_{1}}$. Then the null solution is exponentially
stable, if unique.
\end{remark}

In Fig.\ref{fig}, the piecewise-straight line is the bifurcation
line $\beta=-\beta_{c}(k)$: below it, the system has multiple
stationary solutions, above it, there exists only the null
solution. The curve $\beta=-\bar{\beta}(k)$ is composed of the
straight segment $\beta=\mu_1(k)$, when $0<k\leq\lambda_1$, and
the parabola $\beta=-2\sqrt{k}$, if $k>\lambda_1$.
 It is worth noting that each segment composing the graphic of $-\beta_c$ is tangent to the parabola at $k= \lambda_n=n^4\pi^4$, $n\in\N$.
 In the region of the plain $(\beta,k)$ which is bounded from below by $\beta>-\bar{\beta}(k)$ the exponential stability holds true. Whilst,
in the area between the red and the blue line, when $k>\lambda_1$,
we have asymptotic, but not exponential stability.

\begin{figure}[h]
 \psfrag{O}{\tiny{$O$}}
 \psfrag{k}{\tiny{$k$}}
 \psfrag{b}{\tiny{$\beta$}}
\psfrag{m}{\tiny{$-\sqrt{\lambda_1}$}}
 \psfrag{n}{\tiny{$-2\sqrt{\lambda_1}$}}
 \psfrag{l}{\tiny{$\lambda_1$}}
 \psfrag{p}{\tiny{$4\lambda_1$}}
 \psfrag{q}{\tiny{$\lambda_2$}}
 \psfrag{r}{\tiny{$-5\sqrt{\lambda_1}$}}
 \psfrag{s}{\tiny{$-2\sqrt{\lambda_2}$}}
\begin{center}
\centerline{\includegraphics[width=13.3cm]{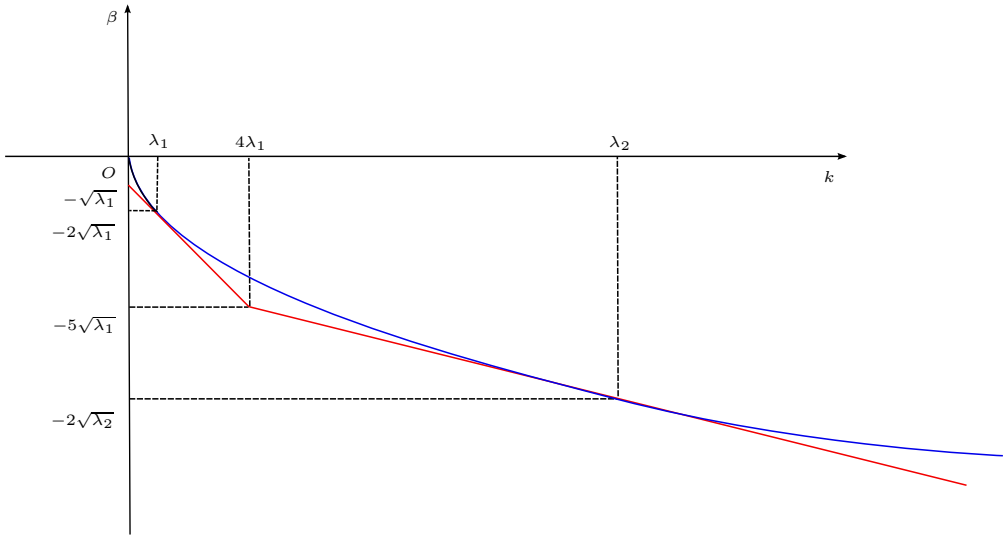}}

\caption{ A picture of the functions $-\bar{\beta}(k)$ and
$-\beta_c(k)$.} \label{fig}
\end{center}
\end{figure}


\section{The Global Attractor}

\noindent We now state the existence of a global attractor for
$S(t)$, for any $\beta \in \mathbb{R}$ and $k\geq0$. Recall that
the global attractor $\A$ is the unique compact subset of $\H_0$
which is at the same time
\begin{itemize}
\item[(i)] attracting:
$$\lim_{t\to\infty}\boldsymbol{\delta}(S(t)\B,\A)\to 0,$$
for every bounded set $\B\subset\H_0$; \item[(ii)] fully
invariant:
$$S(t)\A=\A,$$
for every $t\geq 0$.
\end{itemize}
Here, $\boldsymbol{\delta}$ stands for the Hausdorff semidistance
in $\H_0$, defined as (for $\B_1,\B_2\subset\H_0$)
$$\boldsymbol{\delta}(\B_1,\B_2)=\sup_{z_1\in \B_1}\inf_{z_2\in \B_2}\|z_1-z_2\|_{\H_0}.$$
We address the reader to the books \cite{HAL,TEM} for a detailed
presentation of the theory of attractors. We shall prove the
following

\begin{theorem}
\label{MAIN} The semigroup $S(t)$ on $\H_0$ possesses a connected
global attractor $\A$ bounded in $\H_2=D(A) \times
D(A^{\frac{1}{2}})$. Moreover, $\A$ coincides with the unstable
manifold of the set ${\mathcal S}$ of the stationary states of
$S(t)$, namely,
$$\A=
\left\{
\widetilde{z}\,\,(0):
\begin{array}{cc}
      &  \widetilde{z}\,\,\hbox{is a complete (bounded) trajectory of }S(t):
       \\
 & \underset{t\rightarrow
\infty }{\lim }\left\| \widetilde{z}(-t)-\,S\right\|
_{\H_{0}}=\,0
\end{array}
\right\} .
$$
\end{theorem}

\noindent The set ${\mathcal S}$ consists of all the vectors of the
form $(\bar u,0)$, where $\bar u$ is a (weak) solution to
$$A\bar u+\big(\beta+\|\bar u\|^2_1\big)A^{1/2}\bar u\,+ k\bar u= f.$$
It is then apparent that ${\mathcal S}$ is bounded in $\H_0$. If
${\mathcal S}$ is finite, then
\begin{equation}  \label{A_set}
\A= \left\{
\begin{array}{ccc}
      &  \widetilde{z}\,\,\hbox{is a complete (bounded) trajectory of } S(t)
       \\
  \widetilde{z}\,\,(0):    &  \hbox{such that }\exists\,z_{1}, z_{2}\in S\,:
   \\
     &
      \underset{t\rightarrow \infty }{\lim
}\left\| \widetilde{z}(-t)-\,z_{1}\right\|_{\H_{0}}=\underset{t\rightarrow \infty }{\lim }\left\|\widetilde{z}(t)-\,z_{2}\right\| _{\H_{0}}=\,0
\end{array}
\right\} .
\end{equation}
\begin{remark}
When $f=0$, $k\geq 0$ and $\beta \geq -\beta_c(k)$, then $\A=\mathcal{S}=\{(0,0)\}$. If $\beta < -\beta_c(k)$,
then $\mathcal{S}=\mathcal{S}_0$ may be finite or infinite, according to Theorem \ref{TH-stat.solut}. In the former case, \eqref{A_set} applies.
\end{remark}

The existence of a Lyapunov functional, along with the fact that
$\mathcal{S}$ is a bounded set, allows us prove Theorem \ref{MAIN}
by paralleling some arguments devised in \cite{GPV}.

By the interpolation inequality $\|u\|_1^2\leq \|u\|\|u\|_2$ and
\eqref{POINCARE}, it is clear that
\begin{equation}
\label{gamma} \frac12\|u\|_2^2\leq
\|u\|_2^2+\beta\|u\|_1^2+\gamma\|u\|^2\leq m\|u\|_2^2,
\end{equation}
for some $m=m(\beta,\gamma)\geq 1$,  provided that $\gamma>0$ is large enough.
Now, choosing $\gamma= \mu +k$, where $k>0$ is a fixed value, we assume $\mu$ large enough
so that \eqref{gamma} holds true. Then, we decompose the solution $S(t)z$
into the sum (see \cite{GPV})
$$S(t)z=L(t)z+K(t)z,$$
where
$$L(t)z=(v(t),\pt v(t))\qquad\text{and}\qquad
K(t)z=(w(t),\pt w(t))$$ solve the systems
\begin{equation}
\label{DECAY}
\begin{cases}
\ptt v+  Av+
(\beta+\|u\|^2_1)A^{1/2}v+\mu v + \delta \pt v + k v= 0,\\
\noalign{\vskip1.5mm} (v(0),\pt v(0))=z,
\end{cases}
\end{equation}
and
\begin{equation}
\label{CPT}
\begin{cases}
\ptt w+  Aw+
(\beta+\|u\|^2_1)A^{1/2}w-\mu v + \delta \pt w + k w= f,\\
\noalign{\vskip1.5mm} (w(0),\pt w(0))=0.
\end{cases}
\end{equation}
Having fixed the boundary set of initial data $B(0,R)$, Lemma
\ref{lemma nuovo} entails
$$\underset{t\geq 0}{\sup}\left\{ \left\|
u(t)\right\| _{2}^{2}+\left\| \partial _{t}u(t)\right\|
^{2}\right\} \leq C=C(R)\,,\,\,\,\,\,\,\forall \,z_0\,\in
\,B(0,R)\,,$$
where $\left( u(t)\,,\,\partial _{t}u(t)\right) =\;S(t)z_{0}$;
this formula will be used many times in the subsequent proofs.

In order to parallel the procedure given in \cite {GPV}, we shall
prove Theorem \ref{MAIN} as a consequence of the following five
lemmas. We start showing the exponential decay of $L(t)z$ (see
Lemmas \ref{lemmaDECAY} and \ref{GRNW}) by means of a dissipation
integral (see Lemma \ref{lemmaINT}). Then, we prove the asymptotic
smoothing property of ${K}(t)$, for initial data bounded by $R$
(see Lemma \ref{lemmaCPT}). Finally, by collecting all these
results, we provide the desired consequence (see Lemma
\ref{lemmaABSTRACT}).

Henceforth, let $R>0$ be fixed and $\|z\|_{\H_0}\leq R$.  In addition, $C$ will denote a {\it generic}
positive constant which depends (increasingly) only on $R$, unless
otherwise specified, besides on the structural quantities of the
system.

\begin{lemma}
\label{lemmaDECAY} There is $\omega=\omega(R)>0$ such that
$$\|L(t)z\|_{\H_0}\leq Ce^{-\omega t}.$$
\end{lemma}

\begin{proof}
Denoting
$$\E_0(t)=\|L(t)z\|_{\H_0}^2= \|v(t)\|^2_2+\|\pt v(t)\|^2,$$
for $\eps>0$ to be determined, we set
$$
\Phi_0(t)=\E_0(t)+\beta\|v(t)\|_1^2+ (\mu + k) \|v(t)\|^2
+\|u(t)\|^2_1\|v(t)\|^2_1+\eps\l\pt v(t),v(t)\r.
$$
\newline
In light of Lemma \ref{lemma nuovo} and inequalities (\ref{gamma}), assuming that $\eps$ is small enough, we have the bounds
\begin{equation}
\label{stima-3} \frac14\E_0\leq \Phi_0\leq C\E_0.
\end{equation}
Now, we compute the time-derivative of $\Phi_0$ along the solutions to system (\ref{DECAY}) and we obtain
$$
\frac{d \Phi_0}{dt}+\eps\Phi_0+2(\delta - \eps)\|\pt v\|^2 =2\l\pt
u,A^{1/2}u\r\|v\|^2_1 -\eps(\delta-\eps)\l\pt v,v\r.
$$
Using (\ref{energy}) and assuming $\eps$ small enough (in particular, $\eps<\delta$), we control the rhs
by
$$
C\|\pt u\|\E_0 +\frac{\eps}{8}\E_0 + (\delta - \varepsilon)\|\pt
v(t)\|^2.
$$
So, from (\ref{stima-3}), we obtain
$$
\frac{d \Phi_0}{dt}+\frac{\eps}{2}\Phi_0 +(\delta-\eps)\|\pt v\|^2
\leq C\|\pt u\|\Phi_0.
$$
Finally,
the functional $\Phi_0$
fulfills the differential inequality
\begin{equation} \label{(5.9)}
\frac{d\Phi _{0}}{dt}+\frac{\varepsilon }{2}\Phi_{0}\leq C\left\|
\partial _{t}u\right\| \Phi_{0} .
\end{equation}
The desired conclusion is entailed by applying the following two
lemmas.

\begin{lemma}
\label{GRNW} \rm{(see Lemma 6.2, \cite{GPV})} Let
$\varphi:[0,\infty)\to[0,\infty)$ satisfy
$$
\varphi'+2\eps \varphi\leq g\varphi,
$$
for some $\eps>0$ and some positive function $g$ such that
$$\int_\tau^t g(y)dy\le c_0+\eps (t-\tau),\qquad \forall \tau\in[0,t],$$
with $c_0\geq 0$. Then, there exists $c_1\geq 0$ such that
$$
\varphi(t)\leq c_1\varphi(0)e^{-\eps t}.
$$
\end{lemma}

\begin{lemma} \label{lemmaINT}
For any $\sigma > 0$ small
$$
\int\limits_{\tau }^{t}\left\|
\partial_t u\left( y\right) \right\| dy\leq \sigma (t-\tau
)+\frac{C_{2}}{\sigma }\,\,,
$$
for every $t\geq\tau\geq0.$
\end{lemma}

\begin{proof}
For $\eps \in (0,1]$, we set
$$
\Psi =
\Phi-2 \left\langle f\,,\,u\right\rangle-\frac{1}{2}\beta^2,
$$
where $\Phi$ is defined by \eqref{PHI_FUNZIONALE}.
Taking the time derivative of $\Psi$ and using \eqref{deri-Phi},
we find
\begin{equation} \label{PsiDerivata}
\frac{d\Psi }{dt}+\varepsilon \Psi +\frac{\varepsilon }{2}\left\|
u\right\| _{1}^{4} + 2(\delta - \varepsilon) \|\partial
_{t}u\|^{2} = - \varepsilon \left[ \left\langle
f\,,\,u\right\rangle + (\delta - \varepsilon) \left\langle
\partial _{t}u\,,\,u\right\rangle \right]  .
\end{equation}
By virtue of Lemma \ref{lemma nuovo}, $\E$ is bounded and hence
$$\frac{d\Psi }{dt}+\varepsilon \Psi+ (\delta-\varepsilon)\left\| \partial _{t}u\right\| ^{2} \leq \varepsilon\,C \ .$$
Since $\Psi$
is uniformely bounded by Lemma \ref{lemma_34}, we end up with
$$\frac{d\Psi }{dt}+(\delta-\varepsilon)\left\| \partial _{t}u\right\| ^{2}
\leq \varepsilon \left[ C-\Psi \right] \leq \varepsilon\,C . $$
Integrating this inequality on $[\tau, t]$, we get
$$
(\delta-\varepsilon)\int\limits_{\tau }^{t}\left\| \partial_t
u\left( y\right) \right\|^2 dy\leq \varepsilon\,C(t-\tau )+\Psi
\left( \tau \right) -\Psi \left( t\right) .
$$
Assuming $\varepsilon<\frac{\delta}{2}$, a further application of
Lemma \ref{lemma_34} entails
$$
\int\limits_{\tau }^{t}\left\| \partial_t u\left( y\right)
\right\|^2 dy\leq \varepsilon \,C_{1}(t-\tau )+C_{2} .
$$
Finally, by the H\"{o}lder and the Young's inequalities
$$
\int\limits_{\tau }^{t}\left\|
\partial_t u\left( y\right)
\right\| dy
%
\leq \sqrt{t-\tau } \Big( \varepsilon \,C_{1}(t-\tau )+C_{2}\Big)
^{\frac{1}{2}}\leq
$$
$$
\leq \sqrt{t-\tau }\left( \sqrt{\varepsilon \,C_{1}}\sqrt{t-\tau
}+\sqrt{C_{2}}\right) \leq \sqrt{\varepsilon \,C_{1}}(t-\tau
)+\sqrt{C_{2}}\sqrt{t-\tau } \leq
$$
$$
\leq 2\sqrt{\varepsilon \,C_{1}}(t-\tau
)+\frac{C_{2}}{2\sqrt{\varepsilon \,C_{1}}}=\sigma (t-\tau
)+\frac{C_{2}}{\sigma }
$$
where $\sigma\,=\,2\sqrt{\varepsilon \,C_{1}}$.
\end{proof}
\end{proof}

The next result provides the boundedness of $K(t)z$ in a more
regular space.

\begin{lemma}
\label{lemmaCPT} The estimate
$$\|K(t)z\|_{\H_2}\leq C$$
holds for every $t\geq 0$.
\end{lemma}

\begin{proof}
We  denote
$$\E_1(t)=\|K(t)z\|_{\H_2}^2=\|w(t)\|^2_4+\|\pt w(t)\|_2^2 .$$
For $\eps>0$ small to be fixed later, we set
$$\Phi_1=\E_1
+(\beta+\|u\|^2_1)\|w\|^2_3+\eps\l\pt w,Aw\r-2\l f,A
w\r + k \|w\|^2_2 .$$ The interpolation inequality
$$
\|w\|_3^2\leq\|w\|_2\|w\|_4
$$
and the fact that $\|w\|_2\leq C$ (which follows by comparison
from (\ref{energy}) and Lemma~\ref{lemmaDECAY}) entail
$$\beta\|w\|_3^2\geq-\frac12\E_1-C.$$
Therefore, provided that $\eps$ is small enough, we have the
bounds
$$
\frac12\E_1-C\leq \Phi_1\leq C\E_1+C.
$$
Taking the time-derivative of $\Phi_1$, we find
\begin{align*}
&\frac{d \Phi_1}{dt}+\eps\Phi_1+2(\delta-\eps)\|\pt w\|_2^2=\\
&=2\l\pt u,A^{1/2}u\r\|w\|^2_3+ 2\mu \l A^{1/2}v,A^{1/2}\pt w\r +\\
&\quad+\eps\Big[\mu \l A^{1/2}v,A^{1/2}w\r-(\delta-\eps)\l\pt
w,Aw\r -\l f,Aw\r\Big].
\end{align*}
Using (\ref{energy}) and the above interpolation inequality, if
$\eps$ is small enough, we control the rhs by
$$
\frac\eps8\E_1+C\sqrt{\E_1}\,+C\leq\frac\eps4\E_1
+\frac{C}{\varepsilon} \leq \frac\eps2\Phi_1
+\frac{C}{\varepsilon} .
$$
Hence,
if $\eps$ is fixed small enough (in particular $\eps<\delta$), we obtain
$$
\frac{d \Phi_1}{dt}+\frac\eps2\Phi_1 \leq C.
$$
%
%
Since $\Phi_1(0)=0$, from the Gronwall lemma and the controls
satisfied by $\Phi_1$, we obtain the desired estimate for $\E_1$.
\end{proof}

By collecting previous results, the  following Lemma can be applied to obtain the
existence of the attractor $\A$.

\begin{lemma}
\label{lemmaABSTRACT}\rm{(see \cite{GPV}, Lemma 4.3)} Assume that, for every $R>0$, there exist a
positive function $\psi_R$ vanishing at infinity and a compact set
$\K_R\subset\H_0$ such that the semigroup $S(t)$ can be split into
the sum $L(t)+K(t)$, where the one-parameter operators $L(t)$ and
$K(t)$ fulfill
$$\|L(t)z\|_{\H_0}\leq \psi_R(t)\qquad\text{and}\qquad
K(t)z\in\K_R,$$ whenever $\|z\|_{\H_0}\leq R$ and $t\geq 0$. Then,
$S(t)$ possesses a connected global attractor $\A$, which consists
of the unstable manifold of the set $\mathcal {S}$.
\end{lemma}


\end{document}